\documentclass[graybox]{svmult}

\usepackage{amssymb}

\usepackage{mathptmx}       % selects Times Roman as basic font
\usepackage{helvet}         % selects Helvetica as sans-serif font
\usepackage{courier}        % selects Courier as typewriter font
\usepackage{type1cm}        % activate if the above 3 fonts are
\usepackage{makeidx}         % allows index generation
\usepackage{graphicx}        % standard LaTeX graphics tool
\usepackage{multicol}        % used for the two-column index
\usepackage[bottom]{footmisc}% places footnotes at page bottom

\makeindex             % used for the subject index
\newcommand{\Q}[0]{{\mathbb{Q}}}

\newcommand{\R}[0]{{\mathbb{R}}}

\newcommand{\IR}[0]{{\mathbb{IR}}}

\newcommand{\IQ}[0]{{\mathbb{IQ}}}
\newcommand{\ol}[1]{\mbox{$\overline{{#1}}$}} 
\newcommand{\ul}[1]{\mbox{$\underline{{#1}}$}} 

\newcommand{\diag}{\mathop{\rm diag}\nolimits}

\def\lmin{{\lambda_{\min}}}
\def\lmax{{\lambda_{\max}}}
\newcommand{\pro}[1]{\textsf{#1}}
\newcommand{\size}{\textit{size}}
\newcommand{\NP}{\textrm{NP}}
\newcommand{\pP}{\textrm{P}}
\newcommand{\coNP}{\textrm{coNP}}

\def\eps{{\varepsilon}}

\newcommand{\imace}[1]{\mathbf{#1}}
\newcommand{\smace}[1]{\mathbf{#1}{}^S} 

\definecolor{orange}{RGB}{255,127,0}

\begin{document}

\title{Interval Linear Algebra and Computational Complexity}
% Use \titlerunning{Short Title} for an abbreviated version of
% your contribution title if the original one is too long
\author{Jaroslav Hor\'a\v cek, Milan Hlad\'ik and Michal \v Cern\'y}
% Use \imace{A}uthorrunning{Short Title} for an abbreviated version of
% your contribution title if the original one is too long
\institute{Jaroslav Hor\'a\v cek \at Charles University, Faculty of Mathematics and Physics, Department of Applied
Mathematics, Malostransk\'e n\'am. 25, 118 00, Prague, Czech Republic, \email{horacek@kam.mff.cuni.cz}
\and Milan Hlad\'ik \at Charles University, Faculty of Mathematics and Physics, Department of Applied
Mathematics, Malostransk\'e n\'am. 25, 118 00, Prague, Czech Republic \email{hladik@kam.mff.cuni.cz} \and
 Michal \v Cern\'y \at University of Economics, Faculty of Computer Science and Statistics, n\'am. W. Churchilla 
 4, 13067 Prague, Czech Republic \email{cernym@vse.cz}
}
%
% Use the package "url.sty" to avoid
% problems with special characters
% used in your e-mail or web address
%
\maketitle

\abstract{This work connects two mathematical fields -- computational complexity and interval linear algebra. It introduces the basic topics of interval linear algebra -- regularity and singularity, full column rank, solving a linear system, deciding solvability of a linear system, computing inverse matrix, eigenvalues, checking positive (semi)definiteness or stability. We discuss these problems and relations between them from the view of computational complexity. Many problems in interval linear algebra are intractable, hence we emphasize subclasses of these problems that are easily solvable or decidable. The aim of this work is to provide a basic insight into this field and to provide materials for further reading and research.}

\section{Introduction}
The purpose of this work is to emphasize relations between the two mathematical fields - interval linear algebra and computational complexity. This is not a pioneer work. Variety of relations between interval problems and computational complexity is covered by many papers. There are also few
monographs that are devoted to this topic \cite{fiedler:linopt,kreinovich:complexity,rohn:handbook}. 
Some questions may arise in mind while reading the previous works. Among all, it is the question about the equivalence of the notions
NP-hardness and co-NP-hardness. Some authors use these notions as synonyms. Some distinguish between them. Another 
questions that may arise touches the representation and reducibility of interval problems in a given computational model. We would like to shed more light (not only) on these issues.

Many well-known problems of classical linear algebra become intractable when we introduce intervals into matrices and vectors. However, not everything is lost. 
There are many interesting sub-classes of problems that behave well. We would like to point out these feasible cases, since they are interesting either from the theoretical or the computational point of view.   

Our work does not aspire to replace the classical monographs or handbooks. It lacks many of their details that are cited in the text. Nevertheless, it collects even some recent results that are missing in the monographs. It also provides links and reductions between the various areas of interval linear algebra.  
It provides a necessary and compact introduction to computational complexity and interval linear algebra. Then it considers complexity and feasibility of various  
well-known linear algebraic tasks when considered with interval structures -- regularity and singularity, full column rank, solving a linear system, deciding solvability of a linear system, computing inverse matrix, eigenvalues, checking positive (semi)definiteness or stability. 

We hope this paper should help newcomers to this area to improve her/his orientation in the field or professionals to provide a signpost to more deeper literature.

\section{Interval linear algebra -- part I}
Interval linear algebra is a mathematical field developed from classical linear algebra. The only difference is, that we do not work with real numbers but with real closed intervals 
$$ \imace{a} = [\ul{a}, \ol{a}], $$
where $\ul{a} \leq \ol{a}$.
The set of all closed real intervals is denoted $\IR$ (the set of all closed rational intervals is denoted $\IQ$)
We can use intervals for many reasons -- in applications we sometimes do not know some parameters precisely, that is why, we rather use intervals of possible values; some real numbers are problematic (e.g., $\pi, \sqrt{2}, \ldots$) because it is not easy to represent them precisely, that is why, we can represent them with rigorous intervals containing them etc.
With interval we can define arithmetic (there are more possible definitions, we chose one of the most basic ones). 
\begin{definition}
Let us have two intervals $ \imace{x} = [ \ul{x}, \ol{x} ] $ a $ \imace{y} = [ \ul{y}, \ol{y}] $. The arithmetical operations $ +, *, -, / $ are defined as follows\begin{eqnarray}
\imace{x} + \imace{y} & = & [ \ul{x} + \ul{y}, \ol{x} + \ol{y} ],  \nonumber\\ 
\imace{x} - \imace{y} & =& [ \ul{x} - \ol{y}, \ol{x} - \ul{y} ],  \nonumber\\
\imace{x} * \imace{y}  & = & [ \min(S), \max(S) ], \ \textrm{where} \  S = \{  \ul{xy}, \, \ul{x} \ol{y}, \, \ol{x} \ul{y}, \, \ol{xy} \}, \nonumber\\ 
\imace{x} \, / \, \imace{y}  &= & \imace{x} * (1 / \imace{y}), \quad \textrm{where} \  1/\imace{y} = [1/\ol{y}, 1/\ul{y}], \ 
0 \notin \imace{y}.  \nonumber
\end{eqnarray}
\end{definition} 
Hence, we can use intervals instead of real numbers in formulas. However, we have to be careful. If there is a multiple occurrence of the same interval in a formula, 
the interval arithmetic does see them as two different intervals and we get an overestimation in the resulting interval. For example, let us have 
$\imace{x} = [-2,1]$ and functions $f_1(x) = x^2$ and $f_2(x) = x * x$. Then we get
\begin{eqnarray}
& f_1(\imace{x}) & =  f_1([-2,1]) = [-2,1]^2 = [0,4], \nonumber \\
& f_2(\imace{x}) & =  f_2([-2,1]) = [-2,1] * [-2,1] = [-2,4].  \nonumber 
\end{eqnarray}
In the first case we see the optimal result, in the second case we see overestimation. That is why, the form of our mathematical expression matters. However, we know the cases when the resulting interval is optimal \cite{neumaier:interval}.  
\begin{theorem}
\label{dependencythm}
Applying interval arithmetic on expressions in which all variables occur only once gives the optimal resulting interval. 
\end{theorem}

Using intervals we can build larger structures. In the interval linear algebra the main notion is an interval matrix. It is defined as follows:
$$\imace{A} = \{ A \ | \ \ul{A} \leq A \leq \ol{A} \}, $$
where $\ul{A}, \, \ol{A}$ are real $m \times n$ matrices called \emph{lower} and \emph{upper} bound and the relation $\leq$ is always understood componentwise.   
In another words, it is a matrix with coefficients formed by real closed intervals. In the following text, we will denote every interval structure in boldface. Since an interval vector is a special case of an interval matrix, we define it similarly. We can see that if all intervals in the structures are \emph{degenerate}, i.e, $\ul{A} = \ol{A} $, we get a classical linear algebra. Therefore, interval linear algebra is actually a generalization of the previous one. 

Another way to define an interval matrix is using its \emph{midpoint} matrix $A_c$ and its \emph{radius} matrix $\Delta \geq 0$ as $$\imace{A} = [A_c - \Delta, A_c + \Delta].$$ In the following text we automatically suppose that $A_c, \, \Delta$ represent corresponding midpoint and radius matrix of $\imace{A}$, and 
$b_c, \, \delta$ represent corresponding midpoint and radius vector of $\imace{b}$. 
When we talk about a general square matrix we automatically assume that it is of size $n$. 

We mention some special structures that we will use quite often. The identity matrix is denoted $I$, the matrix containing only ones $E$ and the vector containing only ones $e$. Another useful matrix is $D_y = \diag(y_1, \ldots, y_n)$ a matrix with the vector $y$ as the main diagonal. We often need to describe some properties of interval structures vectors consisting of only $\pm 1$. We denote the set of all $n$-dimensional $\pm 1$ vectors as $Y_n$. 
 A useful concept is a matrix $A_{yz}$ defined as
$$ A_{yz} = A_c - D_y \Delta D_z ,$$
for some given $y,z \in Y_n$. Every its coefficient on the positon $(i,j)$ is an upper or a lower bound of $\imace{A}_{ij}$  depending on the sign of $y_i \cdot z_j$.
We will sometimes need to check spectral radius of a real matrix $A$, we denote it $\varrho(A)$.

Many definitions have an intuitive generalization for interval linear algebra:\\

\emph{An interval matrix $\imace{A}$ has a property $\mathfrak{P}$ if every $ A \in \imace{A}$ has the property $\mathfrak{P}$.} \\

This applies to stability, full column rank, inverse nonnegativity, diagonally dominant matrices, M-matrix and H-matrix property, among others.

Many problems in interval linear algebra are very difficult to be computed exactly (with resulting intervals of tightest possible bounds). 
That is why we inspect the possibility of approximation of these bounds. There are many types of approximation. 
There are several kinds of errors when we approximate a number $a$ -- the absolute, relative \cite{GolLoa1996} and inverse relative \cite{Kre2013} approximation errors.

\begin{definition}
An algorithm computes $a$ with \emph{absolute approximation error} $\eps$ if it computes $a^0$ such that 
%$\lambda^0\in\lmin(\smace{A})+[-\eps,\eps]$.
$a^0\in[a-\eps,\,a +\eps]$.

An algorithm computes $a$ with \emph{relative approximation error} $\eps$ if it computes $a^0$ such that 
%$\lambda^0\in[(1-\eps),(1+\eps)]\lmin(\smace{A})$.
$a^0\in(1+[-\eps,\eps])a$.

An algorithm computes $a$ with \emph{inverse relative approximation error} $\eps$ if it computes $a^0$ such that 
%$\lambda^0\in[(1-\eps),(1+\eps)]\lmin(\smace{A})$.
$a\in(1+[-\eps,\eps])a^0$.
\end{definition}

At the end we mention a very useful theorem that we will use very often in this text. It originally comes from the area of numerical mathematics \cite{oettli1964}.

\begin{theorem}[Oettli-Prager]\label{thmOP}
\label{oettliprager}
Let us have an interval matrix and vector $\imace{A}, \imace{b}$. For a real vector $x \in \R^{n} $ it holds $Ax=b$ for some $A \in \imace{A}, b \in \imace{b}$ if and only if
$$ | A_c x - b_c | \leq \Delta |x| + \delta. $$
\end{theorem}

This was just a brief introduction to interval analysis. 
 Interval linear algebra has many important applications -- system verification, model checking, handling uncertain data. For a huge variety of applications see, e.g., \cite{JauKie2001,kearfott:uses,KeaKre1996}.
 For more information or applications in nonlinear mathematics see \cite{moore:introduction}.

\section{Complexity theory background}
Now, we take a small break and dig deeper into the area of computational complexity. With that in mind we return back to interval linear algebra and introduce some well-known issues from the viewpoint of computational complexity.

\subsection{Binary encoding and size of an instance}

For complexity-theoretic classification of interval-theoretic problems,
it is a standard to use the Turing computation model. We assume that
an instance of a~computational problem is formalized as a bit-string, i.e., a finite 0-1 sequence. Thus we cannot work with real-valued instances; instead we usually restrict ourselves to \emph{rational numbers} expressed as fractions $\pm \frac{q}{r}$ with $q, r \in \mathbb{N}$ written down in binary in the coprime form. Then, the \emph{size} of a rational number $\pm \frac{q}{r}$ is understood as the number of bits necessary to write down the sign and both $q$ and $r$ (to be precise, one should also take care of delimiters).
If an instance of a problem consists of multiple rational numbers $A = (a_1, \dots, a_n)$ (e.g.,~when the input is a vector or a matrix), we define
$\size(A) = \sum_{i=1}^n \size(a_i).$

In interval-theoretic problems, inputs of algorithms are usually interval numbers, vectors or matrices. When we say that an algorithm is to process 
an $m\times n$ interval matrix~$\imace{A}$, we understand that the algorithm is given the pair 
$(\ul{A} \in \mathbb{Q}^{m \times n}, \,
\ol{A} \in \mathbb{Q}^{m \times n})$
and that the size of the input is $L := \size(\ul{A}) + \size(\ol{A})$.
Whenever we speak about \emph{complexity} of such algorithm, we
mean a function $\phi(L)$ counting the number of steps of the corresponding Turing machine as a function of the bit-size $L$ of the input 
$(\ul{A}, \ol{A})$.

Although the literature focuses mainly on
the Turing model (and here we also do so), it is challenging to investigate the behavior
of interval-theoretic problems in other computational models, 
such as the Blum-Shub-Smale (BSS) model for real-valued computing
\cite{BSS} or the quantum model \cite{arorabarak}.

\subsection{Functional problems and decision problems}

Formally, a \emph{functional problem} $\pro{F}$ is
a \emph{total} (defined for each input) function $\pro{F} :\{0,1\}^* \rightarrow \{0,1\}^*$, 
where $\{0,1\}^*$ is the set of all finite bit-strings.
A \emph{decision problem} (or \emph{YES/NO problem}) 
$\pro{A}$ 
is a total function $\pro{A} : \{0,1\}^* \rightarrow \{0,1\}$.

If there exists a Turing machine computing $\pro{A}(x)$ for every $x \in \{0,1\}^*$, we say that the problem $\pro{A}$ (either decision or functional) is \emph{recursive}.

It is well known that many decision problems in mathematics 
are nonrecursive; e.g., deciding whether a given formula is provable in Zermelo-Fraenkel Set Theory is nonrecursive by the famous G\"{o}del Incompleteness Theorem. Fortunately, a majority of decision problems in interval linear algebra are recursive. Such problems can usually be written down as arithmetic formulas (i.e., quantified formulas containing natural number constants, 
arithmetical operations $+, \times$, relations $=, \leq$ and propositional connectives). 
Such formulas are decidable (over the reals) by Tarski's Quantifier Elimination Method \cite{Ren1992a,Ren1992b,Ren1992c}.

\begin{itemize}
\item \emph{Example~A: Regularity of an interval matrix.} Each matrix $A\in\imace{A}$ is nonsingular iff $(\forall A)[\ul{A} \leq A \leq \ol{A} \rightarrow \det(A) \neq 0]$. This formula is arithmetical since $\det(\cdot)$ is a polynomial, and thus it is expressible in terms of $+, \times$.
\item \emph{Example~B: Is a given $\lambda\in \mathbb{Q}$ the largest  eigenvalue of some symmetric $A \in \imace{A}$?} 
This question can be written down as
$
(\exists A)
[A = A^T \ \&\ \ul{A} \leq A \leq \ol{A}
\ \&\ 
(\exists x \neq 0) 
[Ax = \lambda x]
\ \&\ 
(\forall \lambda')
\{(\exists x'\neq 0)
[Ax' = \lambda' x']
\rightarrow \lambda'\leq\lambda\}].
$
\end{itemize}
Although Quantifier Elimination proves recursivity, it is 
a highly inefficient method from the practical viewpoint 
--- the computation time can be doubly exponential in general.
In spite of this, for many problems, 
%such as the problem of Example~B, 
reduction to Quantifier Elimination is the only 
(and thus ``the best'') known algorithmic result.

\subsection{Weak and strong polynomiality}

It is a usual convention to say that a problem $\pro{A}$ is ``efficiently'' solvable if it is solvable in polynomial time, i.e.,~in at most $p(L)$ steps of the corresponding Tu\-ring machine, where $p$ is a polynomial and $L$ is the size of the input. The class of 
efficiently solvable decision problems is denoted by P.

Taking a more detailed viewpoint, this is a definition of polynomial-time solvability in the \emph{weak} sense. 
In our context, we are usually processing a family $a_1, \dots, a_n$ of rational numbers, where
$L = \sum_{i=1}^n \size(a_i)$, performing
arithmetical operations $+,-,\times,\div,\leq$ with them. 
The definition of (weak) polynomiality
implies that an algorithm
\emph{can perform at most $p_1(L)$ arithmetical operations with numbers of size at most $p_2(L)$ during its computation}, where $p_1, p_2$ are polynomials.

If a polynomial-time algorithm satisfies the stronger 
property that it
\emph{performs at most $p_1(n)$ arithmetical operations with numbers of size at most $p_2(L)$ during its computation}, we say that it is \emph{strongly polynomial}.
The difference is whether we can bound the number of arithmetical operations only by a polynomial in $L$, or by a polynomial in $n$.

\textbf{Example.} Given a rational $A$ and $b$, 
the question $(\exists x)[Ax = b]$
can be decided in strongly polynomial time (although it is nontrivial to
implement the Gaussian elimination to yield a strongly polynomial algorithm). 
On the contrary, the question $(\exists x)[Ax \leq b]$ (which is a form of linear programming) is known to be solvable in weakly polynomial time only
and it is a major open question whether a strongly polynomial algorithm exists (this is Smales's Ninth Millenium Problem, see \cite{smale}).

The main message of the previous example is: whenever an interval-algebraic problem is solvable in polynomial time and requires linear programming (which is a frequent case), it is only a weakly polynomial result. This is why the rare cases, when interval-algebraic problems are solvable in strongly polynomial time, are of special interest.
%
%\textbf{Example.} Let $\imace{A},\imace{b}$
%be given and let us decide whether
%\begin{equation}
%(\exists A \in \imace{A})(\exists b\in\imace{b})(\exists x\geq 0)[Ax \leq b].
%\label{eq:mca}
%\end{equation}
%(This problem is known as ``weak nonnegative 
%solvability of the interval inequality
%$\imace{A}x\leq\imace{b}$''.)
%We claim that (\ref{eq:mca}) can be decided (i)~in weakly polynomial time,
%(ii)~with at most $p(\size(\underline{A}))$ arithmetical operations, where $p$ is a polynomial.
%(This is interesting since the number of arithmetical operations \emph{does not} depend on the sizes of $\overline{A}, \underline{b}, \overline{b}$. As a result we get that the problem is strongly polynomial e.g.~for the class of 0-1 matrices $\underline{A}$.) 
%To prove (i),
%it suffices to use Gerlach's Theorem (\ref{??}),
%showing that (\ref{eq:mca}) is equivalent to 
%$(\exists x\geq 0)[\underline{A}x \leq 
%\overline{b}]$, which can be checked in weakly polynomial time by linear programming. This also shows that the computation time does not depend on
%$\overline{A}$ and $\underline{b}$.
%To prove (ii), we must also get rid of $\overline{b}$.
%This follows from
%Tardos' Theorem [ref] (which can be considered as a major achievement in Smale's Ninth Problem):
%\emph{given rational $A$ and $b$, the problem $(\exists x)[Ax\leq b]$ can be solved in polynomial time
%using at most $p_1(\size(A))$ arithmetical operations with numbers 
%of size at most 
%$p_2(\size(A) + \size(b))$ occurring during the computation, where $p_1$ and $p_2$ are polynomials.}
%
%

\subsection{NP, coNP}
\label{orthdecompose} 
Recall that NP is the class of decision problems $\pro{A}$ with the following property:
 there is a polynomial $p$ and a decision problem $\pro{B}(x,y)$, solvable in time polynomial in $\size(x)+\size(y)$, such that, for any 
instance $x\in\{0,1\}^*$, 
\begin{equation}
\pro{A}(x) = 1 \textrm{\ \ iff \ \ } (\exists y \in \{0,1\}^*)\   
\underbrace{\size(y) \leq p(\size(x))}_{(\star)}\ \textrm{\ and\ }\pro{B}(x,y) = 1. 
\label{eq:defNP}
\end{equation}
The string $y$ is called \emph{witness} for the $\exists$-quantifier,
or also \emph{witness} of the fact that $\pro{A}(x) = 1$.
The algorithm for $\pro{B}(x,y)$ is called \emph{verifier}.
For short, we often write 
$
\pro{A}(x) = (\exists^p y)\pro{B}(x,y),
$
showing that $\pro{A}$ results from
the $\exists$-quantification of the efficiently decidable question $\pro{B}$ 
(and the quantifier ranges over strings of polynomially bounded size). Observe
that the question $(\exists^p y)\pro{B}(x,y)$ need not be decidable in polynomial time (in fact, this is the open problem ``$\pP =^{?} \NP$''), since the quantification range is exponential 
in $\size(x)$. 

A lot of $\exists$-problems from various areas of mathematics 
are in NP:
\emph{``does a given boolean formula $x$ 
have a satisfying assignment $y$?''}, 
\emph{``does a given graph $x$ have $3$-coloring $y$?''}, 
\emph{``does a given system $x = \textrm{`}Ay \leq b\textrm{'}$
have an integral solution $y$?''},
and many others.

The class coNP is characterized by replacement of the quantifier in
(\ref{eq:defNP}):
$$ \pro{A}(x) = 1 \textrm{\ \ iff \ \ } (\forall y \in \{0,1\}^*)\ \size(y) \leq p(\size(x)) \rightarrow \pro{B}(x,y) = 1. $$
It is easily seen that the class coNP is formed of complements of NP-problems,
and vice versa. (Recall that a decision problem $\pro{A}$ is a 0-1 function;
its \emph{complement} is defined as $\pro{coA} = 1 - \pro{A}$.)

The prominent example of a coNP-question is deciding whether a boolean formula is a tautology, or in other words, \emph{``given a boolean formula $x$, is it true that \textbf{every} assignment $y$ makes it true\mbox{?}``}.

It is easy to see again that deciding a coNP-question
can take exponential time since the $\forall$-quantifier ranges over a set exponentially large in $\size(x)$.

\textbf{Example.}
Interval linear algebra is not an exception: a lot of $\exists$-questions belong to NP, but we should be careful a bit.
As an example, consider the problem $\pro{SINGULARITY}$: 
given $\imace{A} \in\mathbb{IQ}^{n\times n}$, 
$\exists A \in \imace{A}$ which is singular? We could expect 
that $\pro{SINGULARITY} \in \NP$ since the positive answer 
can be certified
by the $\exists$-witness 
$A_0 = $ \emph{a particular singular matrix in $\imace{A}$}.
Indeed, the natural verifier $\pro{B}(\imace{A},A_0)$, checking whether
$A_0 \in \imace{A}$ and $A_0$ is singular, works in polynomial time.
But a problem is hidden in the condition $(\star)$ in (\ref{eq:defNP}).
To be fully correct, we would have to prove: \emph{there exists a polynomial $p$ such that whenever $\imace{A}$ contains 
a singular matrix, then it also contains a rational 
singular matrix $A_0$ such that $\size(A_0) \leq p(L)$,
where $L = \size(\ul{A}) + \size(\ol{A})$.} Direct proofs of such properties are 
``uncomfortable''. But we can proceed in a more elegant way, using 
Theorem~\ref{oettliprager}: 
\begin{eqnarray}
\nonumber
& \exists A \in & \imace{A} \textrm{\ s.t.~$A$ is singular}  \\ \nonumber
&\Leftrightarrow\ & \exists A \in \imace{A},\ \exists x \neq 0 \textrm{\ s.t.\ } Ax = 0  \\ \nonumber
& \Leftrightarrow\ & \exists x \neq 0 \textrm{\ s.t.\ } -\Delta|x| \leq A_c x \leq \Delta|x|,  \\ \label{eq:expr}
& \Leftrightarrow\ & \exists s\in\{\pm 1\}^n \underbrace{\exists x
\textrm{\ s.t.\ } -\Delta D_s x \leq A_c x \leq \Delta D_s x,\ D_s x \geq 0,\ e^TD_sx\geq 1}_{(\dag)}. 
\end{eqnarray}

Given $s\in\{\pm 1\}^n$, the relation 
(\dag) can be checked in polynomial time by linear programming.
Thus, we can define the verifier $\pro{B}(\imace{A},s)$ as
the algorithm checking the validity of (\dag). 
In fact, we have reformulated the $\exists$-question,
\emph{``is there a singular $A\in\imace{A}$?''}, into an equivalent 
$\exists$-question, \emph{``is there a sign vector $s\in\{\pm 1\}^n$
s.t.~$(\dag)$ holds true?''}, and now  
$\size(s) \leq L$ is obvious. 

The method of (\ref{eq:expr}) is known as \emph{orthant decomposition} since it reduces the problem to inspection of orthants $D_s x \geq 0$, 
for every $s\in\{\pm 1\}^n$, and the work in each orthant is ``easy'' 
(here, the work in an orthant amounts to a single linear program). Many properties with interval data are described by sufficient and necessary conditions that use orthant decomposition. 

We can also immediately see that 
$\pro{REGULARITY} = \pro{coSINGULARITY}$ (\emph{``given $\imace{A}$, is \textbf{every} $A\in \imace{A}$ nonsingular?''}) belongs to coNP.

\subsection{Decision problems: NP-, coNP-completeness}
\label{eigen}

A decision problem $\pro{A}$ is \emph{reducible} to a decision 
problem $\pro{B}$ (denoted $\pro{A} \leq \pro{B}$) if there 
exists a polynomial-time computable function $g: \{0,1\}^* \rightarrow \{0,1\}^*$,
called \emph{reduction}, such that for every $x\in\{0,1\}^*$ we have 
\begin{equation}
\pro{A}(x) = \pro{B}(g(x)). 
\label{eq:mor}
\end{equation}
Said informally, any algorithm for $\pro{B}$ can also be used for solving $\pro{A}$: given an instance $x$ of $\pro{A}$, we can efficiently ``translate'' it into an instance $g(x)$ of the problem $\pro{B}$ and run the method deciding $\pro{B}(g(x))$, yielding the correct answer to $\pro{A}(x)$. Thus, any decision method for $\pro{B}$ is also a valid method for $\pro{A}$, if we admit the polynomial time for computation of the reduction $g$. In this sense we can say that if $\pro{A} \leq \pro{B}$, then $\pro{B}$
``as hard as $\pro{A}$, or harder''. 
If both $\pro{A} \leq \pro{B}$ and $\pro{B} \leq \pro{A}$, then problems $\pro{A}, \pro{B}$ are called \emph{polynomially equivalent}. 

The relation $\leq$ induces a partial ordering on classes of polynomially equivalent problems in $\NP$ (called \emph{NP-degrees}) and this ordering can be shown to have a maximum element. The problems in the maximum class
are called \emph{NP-complete} problems. And similarly, coNP has a class of \emph{coNP-complete} problems. They are complementary: a problem $\pro{A}$ is NP-complete iff its complement is coNP-complete.

Let $\mathcal{X} \in \{\NP,\coNP\}$.
If a problem $\pro{B}$ is $\mathcal{X}$-complete, any method for it can be understood as a universal method for any problem $\pro{A}\in\mathcal{X}$, modulo polynomial time needed for computing the reduction. Indeed, since $\pro{B}$ is the maximum element, we have $\pro{A} \leq \pro{B}$ for any $\pro{A}\in\mathcal{X}$. It is generally believed that $\mathcal{X}$ contains problems which are not efficiently decidable.  
In NP, boolean satisfiability is a prominent example;
in coNP, it is the tautology problem. 
Then, by $\leq$-maximality, no $\mathcal{X}$-complete problem is efficiently decidable. This shows why a proof of $\mathcal{X}$-completeness of a newly studied problem is often understood as proof of its 
computational \emph{intractability}. 

\textbf{Remark.} From a practical perspective, a proof of NP- or coNP-completeness is the same bad news, telling us that ``nothing better than superpolynomial-time algorithms can be expected''.
But formally we must distinguish between NP- and co-NP completeness because it is believed that NP-complete problems are not polynomially equivalent with coNP-complete problems. (This is the 
``NP\ $=^{?}$\ coNP''
open problem). 

\textbf{NP- and coNP-complete problems in interval analysis.} A survey of such problems forms the core of this paper. An important example of an NP-complete problem is \pro{SINGULARITY} of an interval matrix $\imace{A}$.
Its complement, \pro{REGULARITY}, is thus coNP-complete.

When we know that $\pro{B}$ is $\mathcal{X}$-complete and we prove $\pro{B} \leq \pro{C}$ for a problem $\pro{C}\in \mathcal{X}$, then $\pro{C}$ is also $\mathcal{X}$-complete.
This is \emph{the} method behind all $\mathcal{X}$-completeness proofs of this paper.
For example, let \pro{EIGENVALUE} be the problem ``given a square interval matrix $\imace{A}$ and a number $\lambda$, decide whether $\lambda$ is an eigenvalue of some $A\in\imace{A}$''. It is easy to prove
$\pro{SINGULARITY} \leq \pro{EIGENVALUE}$; indeed, if we are to decide
whether there is a singular matrix $A \in \imace{A}$, it suffices to
use the reduction $g : \imace{A} \mapsto (\imace{A}, \lambda = 0)$.
The proof of $\pro{EIGENVALUE}\in\NP$ can be derived from the orthant decomposition method; this proves that \pro{EIGENVALUE} is an NP-complete problem.

\subsection{Decision problems: NP-, coNP-hardness}

We restrict ourselves to NP-hard problems; the reasoning for coNP-hard
problems is analogous.

In the previous section we spoke about NP-complete problems
as the $\leq$-maximum elements in NP. 
But our reasoning can be more general. We can work on the entire class
of decision problems, including those outside NP.
We say that a decision problem $\pro{H}$, not necessarily in NP, satisfying $\pro{C} \leq \pro{H}$ for an NP-complete problem $\pro{C}$, is \emph{NP-hard}. Clearly: NP-complete problems are exactly those NP-hard problems which are in NP. But we might encounter a problem $\pro{H}$ for which we do not have the proof $\pro{H}\in \NP$, but still it might be possible to prove $\pro{C} \leq \pro{H}$. Then the bad news for practice is again the same, that the problem \pro{H} is computationally intractable. 
(But we might possibly need even worse computation time than for NP-problems; recall that all problems in NP can be solved in exponential time, not worse.)

To summarize: a proof that a decision problem is NP-hard is a weaker
theoretical 
result than a proof that a decision problem is NP-complete; it leads to an immediate research problem to inspect \emph{why it is difficult to prove the presence in NP}. Usually, the reason is that it is not easy (or impossible at all) 
to write down the $\exists$-definition; recall the example (\ref{eq:expr}), where
the proof of presence in NP required the aid of Theorem~\ref{oettliprager}.

 \textbf{Remark.} If we are unsuccessful in placing the problem in NP or coNP, being unable to write down the $\exists$- or $\forall$-definition, it might be appropriate to place the problem $\pro{H}$ 
into higher levels of the Polynomial Time Hierarchy, or even higher, such as the PSPACE-level; for details see \cite{arorabarak}, Chapter~5.

%\textbf{Remark.} It might happen that a problem $\pro{H}$ is both NP-hard and coNP-hard (for example, PSPACE-complete problems are); but such problems are likely to be outside $\NP \cup \coNP$, because otherwise $\NP = \coNP$ and this is believed to be false. Thus, when we encounter
%a recursive problem which is both $\NP$-hard and $\coNP$-hard, it is likely to be in a higher level of the Polynomial Time Hierarchy.

\subsection{Functional problems: efficient solvability and NP-hard\-ness}

Functional problems are problems of computing values of general functions, in contrast to decision problems where we expect only YES/NO answers. We also want to classify functional problems from the complexity-theoretic perspective, whether they are ``efficiently solvable'', or ``intractable'', as we did with decision problems. Efficient solvability of a functional problem is again generally understood as polynomial-time computability. 
To define NP-hardness, we need the following notion of reduction: a decision problem $\pro{A}$ is \emph{reducible} to a functional problem \pro{F}, if there exist functions $g : \{0,1\}^* \rightarrow \{0,1\}^*$
and $h: \{0,1\}^* \rightarrow \{0,1\}$, both computable in polynomial time, such that 
\begin{equation}
\pro{A}(x) = h(\pro{F}(g(x))) \ \ \textrm{for every}\ \ x\in\{0,1\}^*. 
\label{eq:rtwo}
\end{equation}
The role of $g$ is analogous to (\ref{eq:mor}): 
it translates an instance $x$ of $\pro{A}$ into an instance $g(x)$ of $\pro{F}$. 
What is new here is the function $h$. Since $\pro{F}$ is a functional problem, the value $\pro{F}(g(x))$ can be an arbitrary bitstring 
(say, a binary representation of a rational number); then we need another efficiently computable function $h$ translating the value $\pro{F}(g(x))$ into a 1-0 value giving the YES/NO answer to $\pro{A}(x)$. A trivial example: deciding regularity of a rational matrix
(decision problem $\pro{A}$) 
is reducible to the computation of rank (functional problem $\pro{F}$). 
It suffices to define $g(A) = A$ and 
$h(\zeta) = 1-\min\{n-\zeta, 1\}$. 

Now, a functional problem $\pro{F}$ is \emph{NP-hard} if there is an NP-hard decision problem reducible to $\pro{F}$. For example, the functional problem 
of counting the number of ones in the truth-table of a given boolean formula is NP-hard since this information allows us to decide whether or not the formula is satisfiable.

 \textbf{Remark. It is not necessary to distinguish between NP-hardness and coNP-hardness for functional problems.} 
We could also try to define coNP-hardness of a functional problem
$\pro{G}$ in terms of reducibility of a coNP-hard decision 
problem $\pro{C}$ to $\pro{G}$ via (\ref{eq:rtwo}).
But this is superfluous because here NP-hardness and coNP-hardness coincide.
Indeed, if we can reduce a coNP-hard problem $\pro{C}$ to a functional problem $\pro{G}$ via $(g,h)$, then we can also reduce the NP-hard problem $\pro{coC}$ to $\pro{G}$ via $(g,1-h)$. Thus, in case of functional problems, we speak about NP-hardness only. 

%\textbf{Remark.} We do not have a direct analogy of NP and coNP for functional problems; thus we do not speak about ``completeness'' of functional problems.
%(In spite of this, various definitions in logic are available, such as $\Sigma_1$- and $\Pi_1$-definable functions; but this would lead us to a different area with no particular gain in interval-algebraic problems.)

\subsection{More general reductions: 
do we indeed have to distinguish between NP-hardness 
and coNP-hardness of decision problems?} 
In literature, the notions of NP-hardness and coNP-hardness
are sometimes used quite freely even for \emph{decision problems}. Sometimes
we can read that a decision problem is ``NP-hard'', even if it would qualify as a coNP-hard problem under our definition based on the reduction 
(\ref{eq:mor}). This is nothing serious as far as we are aware. It depends how the author understands the notion of a reduction between two decision problems. We have used the \emph{many-one} reduction (\ref{eq:mor}), known also as \emph{Karp} reduction, between two decision problems. This is a standard in complexity-theoretic literature. 

However, one could use a more general reduction between two decision problems $\pro{A}, \pro{B}$. For example, taking inspiration from
$(\ref{eq:rtwo})$,
we could define ``$\pro{A} \leq' \pro{B}$ iff $\pro{A}(x) = h(\pro{B}(g(x)))$ for some polynomial-time computable functions $g,h$''. Then the notions of $\leq'$-NP-hardness and $\leq'$-coNP-hardness coincide and need not be distinguished. 
(Observe that $h$ must be a function from $\{0,1\}$ to $\{0,1\}$ and there are only two such nonconstant functions: $h_1(\xi) = \xi$ and $h_2(\xi) = 1-\xi$. If we admit only $h_1$, we get the many-one reduction; if we admit also the negation $h_2$, 
we have a generalized reduction
under which a problem is NP-hard iff it is coNP-hard.
Thus: the notions of NP-hardness and coNP-hardness based on many-one reductions do not coincide just because many-one reductions \emph{do not admit the negation} of the output of $\pro{B}(g(x))$.) 

To be fully precise, one should always say ``a problem $\pro{A}$ is $\mathcal{X}$-hard w.r.t.~a particular reduction $\preceq$''. For example, in the previous sections we spoke 
about $\mathcal{X}$-hard problems for $\mathcal{X} \in \{\NP,\coNP\}$
w.r.t.~the many-one reduction (\ref{eq:mor}). If another author uses 
$\mathcal{X}$-hardness w.r.t.~$\leq'$ (e.g., ~because (s)he considers the ban of negation as too restrictive in 
her/his context), 
then (s)he need not distinguish between NP-hardness and coNP-hardness. 

For the sake of completeness, we conclude that in literature we can meet the notions of hardness w.r.t.~various types of reductions.

\emph{Logspace-computable reduction:} $\pro{A} \leq_{\log} \pro{B}$ iff there is a function $g$ computable in memory of size $O(\log\size(x))$, such that $\pro{A}(x) = \pro{B}(g(x))$ for every $x$. (This reduction in weaker than (\ref{eq:mor}) since every logspace-computable function is also computable in polynomial time.)

\emph{Truth-table reduction:} $\pro{A} \leq_{tt} \pro{B}$ 
iff there is a finite number of polynomial-time computable functions
$g_1, \dots, g_k : \{0,1\}^* \rightarrow \{0,1\}^*$ 
and a ``truth-table'' function $h : \{0,1\}^k \rightarrow \{0,1\}$ such that
$\pro{A}(x) = h(\pro{B}(g_1(x)), \dots, \pro{B}(g_k(x)))$. This reduction is a generalization of $\leq'$; indeed, $\leq'$ is a restricted truth-table 
reduction with a two-line truth table. 
Under $\leq_{tt}$, to decide $\pro{A}(x)$ one can compute $k$ instances of $\pro{B}$ from which the boolean expression $h$ combines the result $\pro{A}(x)$. 

\emph{Turing reduction} (or \emph{Cook reduction}): $\pro{A} \leq_T \pro{B}$ iff there is a polynomial-time algorithm (Turing machine) $Q$, equipped with a subroutine (an algorithm, \emph{oracle}) computing 
$\pro{B}$, and the entire computation of $\pro{B}$ is counted as a single step of $Q$. This is the most general type of reduction: when deciding $\pro{A}(x)$, the reduction allows for a polynomial number of computations of $\pro{B}(y)$
with $\size(y)$ polynomially bounded in $\size(x)$, and the results can be combined in an arbitrary way; the only limitation is that the overall number of steps is polynomial in $\size(x)$, assuming that one computation of $\pro{B}(y)$ is at the unit cost. 

The above mentioned reductions can be ordered in the sequence 
according to their generality:
$\pro{A} \leq_{\log} \pro{B} \Rightarrow
\pro{A} \leq \pro{B} \Rightarrow 
\pro{A} \leq' \pro{B} \Rightarrow
\pro{A} \leq_{tt} \pro{B} \Rightarrow
\pro{A} \leq_{T} \pro{B}$,
where ``$\Rightarrow$'' means ``implies''.
We know that NP-hardness and coNP-hardness coincide for $\leq'$, and thus also for the generalizations $\leq_{tt}$, $\leq_T$.

\subsection{A reduction-free definition of hardness}

For practical purposes, when we do not want to play with properties of particular reductions, we can define the notion of a ``hard'' problem $\pro{H}$ (either decision of functional) intuitively as a problem fulfilling this implication: \emph{if $\pro{H}$ is decidable/solvable in polynomial time, then $\pP = \NP$}. This is usually satisfactory for the practical understanding of the notion of computational hardness. (Under this definition: if $\pP = \NP$, then every decision problem is hard; and if $\pP \neq \NP$, then the class of hard decision problems is exactly the class of decision problems not decidable in polynomial time, including all NP-hard and coNP-hard decision problems.)

Even if we accept this definition and do not speak about reductions explicitly, 
all hardness proofs (at least implicitly) contain
some kinds of reductions of previously known hard problems to the 
newly studied ones.

\section{Interval linear algebra -- part II}
In the following sections we will deal with various problems in interval linear algebra. There are many interesting topics that are unfortunately beyond the scope of this work. We will at least point out some of them in section \ref{further}. We chose basic topics from introductory courses to linear algebra -- regularity and singularity of a matrix, full column rank, solving and solvability of a system of linear equations, matrix inverse, determinant, eigenvalues and eigenvectors, positive (semi)definiteness and stability. The next chapters will offer a great disappointment and also a great challenge, since implanting intervals into a classical linear algebra makes solving most of the problems intractable. That is why, we look for solving relaxed problems, special feasible subclasses of problems or for sufficient conditions checkable in polynomial time. Interval linear algebra still offers many open problems and a lot of place for further research. At the end of each section we present a summary of problems and their complexity. If we only know that a problem is weakly polynomial yet, we just write that it belongs to the class $P$. When complexity of a problem is not known to our best knowledge (or it is an open problem), we mark it with question mark.

%================================================================================================
\subsection{Regularity and singularity}  
 
Deciding regularity and singularity of an interval matrix is an important task in linear algebra . The definition of interval regularity (and singularity) is intuitive. 
\begin{definition}
A square interval matrix $\imace{A}$ is \emph{regular} if every $A \in \imace{A}$ is nonsingular. Otherwise, $\imace{A}$ is called \emph{singular}.
\end{definition}  
Considering complexity we can find in the literature the following theorem \cite{rohn:checking} giving NP-completeness result even for the simple case. 
\begin{theorem}
Deciding whether an interval matrix $\imace{A} = [A - E, A + E]$ is singular for some nonnegative symmetric positive definite rational matrix $A$ is NP-complete.
\end{theorem}
We can prove NP-hardness of this decision problem. Moreover, we get NP-complete\-ness since we know that a singular $\imace{A}$ in this form mentioned in the theorem must contain a singular matrix
$$ A - \frac{zz^T}{z^T A^{-1} z},$$
for some $z \in \{\pm 1\}^n$ \cite{rohn:checking} which is a polynomial witness and the above mentioned matrix is checkable in polynomial time (e.g., by Gaussian elimination). This implies that deciding singularity of a general interval matrix is NP-hard. However, in the section \ref{orthdecompose} we saw the construction of a polynomial witness $z \in \{\pm 1\}^n$ certifying that an interval matrix is singular. Hence, we get that checking singularity of a general interval matrix is NP-complete. Clearly, checking regularity as the complement problem to singularity is coNP-complete.  

The sufficient and necessary conditions for checking regularity are of exponential nature. In \cite{rohn:forty} you can see 40 of them. 
For example, we can use the classical definition of matrix regularity (a matrix $A$ is regular if the system $Ax=0$ has only trivial solution) and combine it with Oettli-Prager theorem. We get that an interval matrix is regular if and only if the inequality $$ |A_c x| \leq \Delta|x|, $$
has only trivial solution. 

Fortunately, there are some sufficient conditions that are computable in polynomial time. 
It is advantageous to have more conditions, because some of them may suit better to a certain class of matrices or limits of our software tools. Here we present three sufficient conditions for checking regularity and three sufficient conditions for checking singularity.

\begin{theorem}[Sufficient conditions for regularity]
An interval matrix $\imace{A} = [A_c - \Delta, A_c + \Delta]$ is regular if at least one of the following conditions holds
\begin{enumerate}
\item $ \varrho(|A^{-1}_c|\Delta) < 1$ \ \cite{rohn:checking}, \vspace{3pt}
\item $ \sigma_{\max}(\Delta) < \sigma_{\min}(A_c)$ \ \cite{rump1994verification}, \vspace{3pt}
\item $ A^T_c A_c - \| \Delta^T \Delta\| I$ is positive definite for some consistent matrix norm $\|\cdot\|$ \ \cite{rex:sufficient}.
\end{enumerate}

\end{theorem}

\begin{theorem}[Sufficient conditions for singularity]
An interval matrix $\imace{A} = [A_c - \Delta, A_c + \Delta]$ is singular if at least one of the following conditions holds
\begin{enumerate}
\item $ \max_j(|A^{-1}_c| \Delta)_{jj} \geq 1$ \ \cite{rohn1989systems}, \vspace{3pt}
\item $ (\Delta - |A_c|)^{-1} \geq 0$ \ \cite{rohn:checking},\vspace{3pt}
\item $ \Delta^T \Delta\ - A^T_c A_c $ is positive semidefinite \ \cite{rex:sufficient}.
\end{enumerate}
\end{theorem}

In the above two theorems, the first condition in the triplet is among the most frequently used sufficient conditions. You can find more sufficient conditions for regularity and singularity in \cite{rex:sufficient}. 

We can also take a look at the classes of interval matrices that are immediately regular. These are, for example, diagonally dominant matrices \cite{sharyfcr}, M-matrices and H-matrices \cite{neumaier:interval}. There properties are checkable in polynomial time.\\

\noindent \textbf{Summary.}
\begin{center}\renewcommand\arraystretch{1.3} 
\tabcolsep=1em
\begin{tabular}{@{}ll@{}}
\svhline
 \emph{Problem} & \emph{Complexity} \\
\hline
Is $\imace{A}$ regular?  &  coNP-complete  \\
Is $\imace{A}$ singular?  & NP-complete  \\
\svhline
\end{tabular}
\end{center}

%================================================================================================
\subsection{Full column rank}

The definition of the full column rank is natural. 

\begin{definition}
An $m \times n$ interval matrix $\imace{A}$ has \emph{full column rank} if every $A \in \imace{A}$ has full column rank (i.e., it has rank $n$). 
\end{definition}

Deciding whether an interval matrix has full column rank is connected to checking regularity.
If an interval matrix $\imace{A}$ of size $m\times n$, $m\geq n$, contains a regular submatrix of size $n$, then obviously $\imace{A}$ has a full column rank.
What is surprising is that the implication does not hold conversely (in contrast to real matrices). The interval matrix by Irene Sharaya (see \cite{sharyfcr}) might serve as a counterexample.
$$
\imace{A}=\left(\begin{array}{cc}
    $1$  & $[0,1]$\\
   $-1$  & $[0,1]$\\
$[-1,1]$ &  $1$
\end{array}\right).
$$
It has full column rank, but contains no regular submatrix of size 2.

For square matrices, checking regularity can can be polynomially reduced to checking full column rank (we just check the matrix $\imace{A}$), but the converse is not so easy. Therefore, checking full column rank is coNP-hard. Finding a polynomial certificate for an interval matrix not having full column rank can be done by orthant decomposition similarly as in the case of singularity. That is why, checking full column rank if coNP-complete.

Again, fortunately, we have some sufficient conditions that are computable in polynomial time. 

\begin{theorem} 
Let $\imace{A} = [A_c - \Delta, A_c + \Delta] $ be an $m \times n$ interval matrix. This matrix has full column rank if at least one of the following conditions holds 

\begin{enumerate}
\item $A_c$ has full column rank and  $\varrho ( | A^\dag_c | \Delta ) < 1$, \cite{rohn:handbook},
%\item $\exists u >0$, $\exists R$ such that $\| I - R\imace{A}\|_u < 1$ 
\item $ \sigma_{\max}(\Delta) < \sigma_{\min}(A_c)$, \cite{sharyfcr}.
%\item $\|\Delta\| < \|A^+_c\|^{-1}$ \cite{sharyfcr}
\end{enumerate}
\end{theorem}
%The second condition is our generalization of traditional matrix theorem. We can use it for the matrix $R \sim (A^c)^+$. For square matrices, the first condition is %actually the same as the sufficient regularity condition. Also, for square matrices it can be proved that these conditions are equivalent [Neumaier]. We strongly %believe that they are equivalent also for rectangular matrices (the testing implies so), but yet we do not have any proof. The third condition is by Shary.
The symbol $^\dag$ stands for Moore-Penrose inverse. The first condition is mentioned implicitly in \cite{rohn:handbook}, however the explicit proof can be found in \cite{sharyfcr}. 
Notice that the second sufficient condition is the same as the sufficient condition for checking regu\-larity.       
Many problems can be transformed to checking full column rank -- e.g., deciding whether a given interval linear system is solvable, deciding whether a solution set of an interval linear system is bounded.\\

\noindent \textbf{Summary.}
\begin{center}\renewcommand\arraystretch{1.3} 
\tabcolsep=1em
\begin{tabular}{@{}ll@{}}
\svhline
 \emph{Problem} & \emph{Complexity} \\
\hline
Does $\imace{A}$ have full column rank? &  coNP-complete  \\
\svhline
\end{tabular}
\end{center}

%================================================================================================
\subsection{Solving a system of linear equations}

To be brief the title of this section contained the word "solving". Nevertheless, this notion could be a little misguiding. 
Let us explain what do we mean by solving a system of interval linear equations (or interval linear system for short).
The solution set of an interval linear system is defined as follows.
\begin{definition}\label{dfSol}
Let $\imace{A} x = \imace{b}$, where $\imace{A}$ is an $m \times n$ interval matrix and $\imace{b}$ is an $m$-dimensional right-hand side vector. Then by a \emph{solution set} $\Sigma$ we mean
$$ \Sigma = \{ x \ | \ Ax=b \ \textrm{for some} \ A \in \imace{A}, \ b \in \imace{b} \}.$$
\end{definition}
We could imagine it as a collection of all solutions of all crisp real systems contained within the bounds of an interval system. Unfortunately, this set is of quite a complex shape. For its description we can use the already mentioned Oettli-Prager Theorem~\ref{thmOP}. A vector $x \in \R^{n} $ is a \emph{solution} of $\imace{A} x = \imace{b}$ (i.e., $x \in \Sigma$) if and only if $x$ satisfies
$$ | A_c x - b_c | \leq \Delta |x| + \delta. $$
We can see that checking whether a vector $y$ is a solution of $\imace{A} x = \imace{b}$ is strongly polynomial (we just check the inequality for $y$).

Oettli-Prager theorem implies that the set $\Sigma$ is generally non-convex but convex in each orthant (for graphical examples of possible shapes of the solution set see e.g., \cite{moore:introduction, horacek:oils, neumaier:ils}). That is why, we usually approximate this set by an $n$-dimensional box (aligned with axes) containing $\Sigma$. Notice that we can view an $n$-dimensional interval vector as an $n$-dimensional box aligned with axes.

\begin{definition}
An $n$-dimensional interval vector $\imace{x}$ is called an interval \emph{enclosure} of $\Sigma$ if $\Sigma \subseteq \imace{x}$. If it is the tightest possible enclosure w.r.t.\ inclusion (there is no interval box $\imace{y}$ such that $\Sigma \subseteq \imace{y} \subsetneqq \imace{x} $), we call $\imace{x}$ the interval \emph{hull}. 
\end{definition}

By \emph{solving} an interval linear system we understand computing any enclosure $\imace{x}$ of its solution set $\Sigma$. To be brief, we call that $\imace{x}$ an enclosure (or the hull) of $\imace{A} x = \imace{b}$.
The notion of enclosure is quite intuitive because we are not always able to compute the interval hull. In \cite{kreinovich:complexity} we can see that computing the exact hull of $\imace{A} x = \imace{b}$ is $\NP$-hard.

An interval $\imace{a} =  [a - \Delta, a + \Delta]$ is absolutely $\delta$-narrow if $\Delta \leq \delta$ and relatively $\delta$-narrow if $\Delta \leq \delta \cdot |a|$.
The problem is still NP-hard even if we limit widths of intervals of a matrix in a system with some $\delta > 0$ \cite{kreinovich:complexity}. We can summarize it in the following theorem. 
\begin{theorem}
For every $\delta > 0$, the problem of computing the hull of $\imace{A} x = \imace{b}$, where $\imace{a}_{ij}, \imace{b}_i$ are both absolutely and relatively $\delta$-narrow is \NP-hard. 
\end{theorem}

Unfortunately, even computing various $\eps$-approximations of the hull components is an $\NP$-hard problem \cite{kreinovich:complexity}. 

\begin{theorem}
For a given $\eps > 0$ computing the relative and absolute $\eps$-approximation of the hull (its components) of $\imace{A} x = \imace{b}$ are NP-hard problems.
\end{theorem}

%If we are interested in only symmetric matrices from a certain system $\A x=\B$, we have to slightly adapt the definitions of $\Sigma$ (which we will not do here). %However, all the mentioned complexity results hold even for the symmetric case. \\

That is why, we are usually looking for enclosures, not the hull. Of course,  the tighter enclosure the better. 
For computing enclosures of square systems, there have been various methods developed. Some of them extend the traditional algorithms for the real systems, such as  the Gaussian elimination, Jacobi or Gauss-Seidel method \cite{moore:introduction,neumaier:interval}. Some of them were designed specifically for interval systems; see for instance \cite{Hla2014b,hladik:shaving,krawczyk:newton,moore:introduction,neumaier:interval,fiedler:linopt} among many others.\\

\noindent \textbf{Overdetermined systems.}
For an \emph{overdetermined} system (where $\imace{A}$ is an $m \times n$ matrix with $m > n$) the situation is slightly more difficult. Many people automatically think of solving overdetermined systems via least squares, i.e., 
\begin{definition}
$$\Sigma^{lsq} = \{ x \ | \ A^TAx=A^Tb \ \textrm{for some} \ A \in \imace{A}, b \in \imace{b} \}.$$
\end{definition}
Obviously, $\Sigma^{lsq}$ is not the same set as $\Sigma$. Nevertheless, it is not difficult to see that $\Sigma \subseteq \Sigma^{lsq}$. Hence, we can use methods for computing least squares for enclosing $\Sigma$ \cite{neumaier:ils}. The problem of computing the interval hull of $\Sigma^{lsq}$ is NP-hard, since when $\imace{A}$ is square and regular, then $\Sigma^{lsq} = \Sigma$ and computing the exact hull of $\Sigma$ is NP-hard even for $\imace{A}$ regular \cite{fiedler:linopt}.

If we primarily focus on enclosing just $\Sigma$ there is a variety of methods -- modified Gaussian elimination for overdetermined systems \cite{hansen:ols} , method developed by Rohn \cite{rohn:enclosing}, Popova \cite{popova:oils}, or a method using square subsystems \cite{horacek:subsq}. 
%
%When we wish to compute the hull of an interval linear system, we can use linear programming. If we take a look at nonlinear Oettli-Prager formula and restrict ourselves to only one orthant that we get a system of linear inequalities, so we can use (verified) linear programming now. We need to use $2^n \times 2n$ linear programs ($2n$ for each orthant to find maximum and minimum in each direction). However, if we precompute some tight enclosure $\X$ this might give us a hint about the resulting signs of $\X$ (or we might know the signs from the nature of the problem) and reduce heavily the number of linear programs used. 
%\MH{ma to nejaky vyznam?}
%It is clear that if we are looking only for positive solution to our system, then the problem lies in P (we are checking only one orthant). \\ 

We can try to identify some classes of systems with exact hull computation algorithms that run in polynomial time.  If we restrict the right hand side $\imace{b}$ to contain only degenerate intervals, we have $\imace{A} x = b$. Then, this problem is still NP-hard \cite{kreinovich:complexity}. If we, however, restricts the matrix to be consisting only of degenerate intervals $A$ and we have a system $A x = \imace{b}$, then, computing exact bounds of the solution set is polynomial, since it can be rewritten as a linear program.

However, even if we allow at most one nondegenerate interval coefficient in each equation, the problem becomes again NP-hard, since an arbitrary interval linear system can be rewritten in this form \cite{kreinovich:complexity}. \\

\noindent \textbf{Structured systems.}
We can also explore band and sparse matrices. 
\begin{definition}
A matrix $\imace{A}$ is a $w$-\emph{band} matrix if $\imace{a}_{ij}=0$ for $|i-j| \geq w$. 
\end{definition}

Band matrices with $d=1$ are diagonal and computing the hull is clearly strongly polynomial. For $d=2$ (tridiagonal matrix) it is an open problem. And for $d\geq 3$ it is again NP-hard. 
We inspected the case of bidiagonal matrices. The result is to our best knowledge new. 

\begin{theorem}
For a bidiagonal matrix (the matrix with only the main diagonal and an arbitrary neighbouring diagonal) computing the exact hull of $\imace{A} x = \imace{b}$ is strongly polynomial. 
\end{theorem}

\begin{proof}
Without the loss of generality let us suppose that the matrix $\imace{A}$ consists of the main diagonal and the one beyond it. 
By the forward substitution, we have $\imace{x}_1=\frac{\imace{b}_1}{\imace{a}_{11}}$ and
$$
\imace{x}_i
=\frac{\imace{b}_i-\imace{a}_{i,i-1}\imace{x}_{i-1}}{\imace{a}_{ii}},
\quad i=2,\dots,n.
$$
By induction, $\imace{x}_{i-1}$ is optimally computed with no use of interval coefficients of the $i$th equations. Since an evaluation in interval arithmetic is optimal in the case there are no multiple occurrences of variables (Theorem \ref{dependencythm}), $\imace{x}_{i}$ is optimal as well.
\end{proof}

\begin{definition}
A matrix $\imace{A}$ is a $d$-\emph{sparse} matrix if in each row $i$ at most $d$ elements $\imace{a}_{ij} \neq 0$. 
\end{definition}

For sparse matrices with $d=1$ computing the hull is clearly strongly polynomial. For $d \geq 2$ it is again NP-hard \cite{kreinovich:complexity}. Nevertheless, if we combine w-band matrix with system coefficient bounds coming from a given finite set of rational numbers, then we have a polynomial algorithm for computing the hull \cite{kreinovich:complexity}.   

If an interval system $\imace{A} x = \imace{b}$  is in a certain form the hull can be computed in polynomial time using some already introduced algorithms. 
If the matrix $\imace{A}$ has full column rank and $A_c$ is a diagonal matrix with positive entries, then Hansen-Bliek-Rohn prescription for enclosure gives the exact hull \cite{fiedler:linopt}. If $\imace{A}$ is an M-matrix, then Gauss-Seidel iteration method converges to the exact hull \cite{neumaier:interval}. And if $\imace{A}$ is an M-matrix and $\imace{b}$ is nonnegative then the interval version of Gaussian elimination yields the exact hull \cite{neumaier:interval}.

In this section we silently supposed that the solution set $\Sigma$ is bounded. This is not always the case. Many mentioned methods can not deal with an unbounded solution set. That is why we usually need to check for boundedness. However, it is an coNP-complete problem since it is identical with checking the full column rank of the interval matrix $\imace{A}$. 

 \textbf{Remark.}
A natural generalization of an interval linear system is by incorporating linear dependencies. That is, we have a family of linear systems
\begin{equation} 
\label{sysPar}
A(p)x=b(p),\quad p\in\imace{p},
\end{equation}
where $A(p)=\sum_{k=1}^KA^kp_k$ and  $b(p)=\sum_{k=1}^Kb^kp_k$. Here, $p$ is a vector of parameters varying in $\imace{p}$. 
Since this concept generalizes the standard interval systems, many related problems are intractable. We point out one particular efficiently solvable problem. Given $x\in\R^n$, deciding whether it is a solution of a standard interval system $\imace{A} x=\imace{b}$ is strongly polynomial. For systems with linear dependencies, the problem still stays polynomial, but we can show weak polynomiality only; this is achieved by rewriting (\ref{sysPar}) as a linear program. \\

\noindent \textbf{Summary.}
\begin{center}\renewcommand\arraystretch{1.3} 
\tabcolsep=1em
\begin{tabular}{@{}ll@{}}
\svhline
 \emph{Problem} & \emph{Complexity} \\
\hline
Is $x$ a solution of $\imace{A} x = \imace{b}$? $\quad$ &   strongly P  \\[0.4ex] 
Computing an enclosure of $\imace{A} x = \imace{b}$  &   P  \\
Computing the hull of  $\imace{A} x = \imace{b}$ &   NP-hard  \\
Computing the hull of $\imace{A} x = b$ & NP-hard \\
Computing the hull of $A x = \imace{b}$ & P \\
Computing the hull of $\imace{A} x = \imace{b}$, where $\imace{A}$ is regular & NP-hard \\
Computing the hull of $\imace{A} x = \imace{b}$, where $\imace{A}$ is M-matrix & P \\
Computing the hull of $\imace{A} x = \imace{b}$, where $\imace{A}$ is diagonal &  strongly P \\
Computing the hull of $\imace{A} x = \imace{b}$, where $\imace{A}$ is bidiagonal &  strongly P \\
Computing the hull of $\imace{A} x = \imace{b}$, where $\imace{A}$ is tridiagonal &  ? \\
Computing the hull of $\imace{A} x = \imace{b}$, where $\imace{A}$ is 3-band &  NP-hard \\
Computing the hull of $\imace{A} x = \imace{b}$, where $\imace{A}$ is 1-sparse &  strongly P \\
Computing the hull of $\imace{A} x = \imace{b}$, where $\imace{A}$ is 2-sparse &  NP-hard \\
Computing the exact least squares hull of $\imace{A} x = \imace{b}$  &  NP-hard \\
Is $\Sigma$ bounded? & coNP-complete \\
\svhline
\end{tabular}
\end{center}

%================================================================================================
\subsection{Matrix inverse} 

Computation of a matrix inverse is usually avoided in applications. Nonetheless, we chose to mention this topic, since it holds a worthy place in interval linear algebra theory. An interval inverse matrix is defined as follows.

\begin{definition}
Let us have a square regular interval matrix $\imace{A}$. We define its interval inverse matrix as $\imace{A}^{-1} = [\ul{B}, \, \ol{B}]$, where $\ul{B} = \min\{A^{-1}, \ A \in \imace{A}\}$ and $\ol{B} = \max\{A^{-1}, \ A \in \imace{A}\}$, where the $\min$ and $\max$ is understood componentwise.
\end{definition}

As usual, the inverse matrix can be computed using knowledge of inverses of boundary matrices $A_{yz}$ \cite{rohn1993inverse}.

\begin{theorem}
Let $\imace{A}$ be regular. Then its inverse $\imace{A}^{-1} = [\ul{B}, \, \ol{B}]$ is described by
$$ \ul{B} = \min_{y,\,z \in Y_n } A^{-1}_{yz},$$
$$ \ol{B} = \max_{y,\,z \in Y_n} A^{-1}_{yz},$$
where the $\min$ and $\max$ is understood componentwise.
\end{theorem}

The maximum and minimum bound of each component of the interval inverse is attained at one of the inverse of $2^{2n}$ boundary matrices.
No wonder, it can be proved that generally computing exact inverse matrix is NP-hard \cite{coxson1999computing}.

When $A_c = I$, we can compute the exact inverse in polynomial time according to the next theorem \cite{rohn2011explicit}. 

\begin{theorem}
Let $\imace{A}$ be a regular interval matrix with $A_c = I$. Let $M = (I - \Delta)^{-1}$. Then its inverse $\imace{A}^{-1} = [\ul{B}, \, \ol{B}]$ is described by
\begin{eqnarray}
 \ul{B} & = &  -M + D_k, \nonumber\\
\ol{B} & = & -M,  \nonumber
\end{eqnarray}
where $ k_j = \frac{2m^2_{jj}}{2m_{jj} - 1} \ $ for $j =  1, \ldots, n$, with $m_{jj}$ being diagonal elements of $M$.
\end{theorem}
There also exists a formula for the exact matrix inverse if all intervals have uniform widths, i.e., $\imace{A} = [A_c - \alpha E, A_c + \alpha E]$ \cite{rohn2011inverse}. 

If we wish to only compute an enclosure $\imace{B}$ of the matrix inverse we can use any method for computing enclosures of interval linear systems. 
 We get the $i$-th column of $\imace{B}$ by solving the systems $\imace{A} x = e_i$, where $e_i$ is $i$-th column of the identity matrix of order $n$. 

As we mentioned, computing the exact interval inverse is NP-hard. 
We close this section with a surprising result on inverse nonnegativity (${A}^{-1}\geq0$ for every $A\in\imace{A}$). It was first proved in slightly different form in \cite{kuttler}. For this form see \cite{neumaier:interval}. It implies that checking inverse nonnegativity and also computing the exact interval inverse of an inverse nonnegative matrix is strongly polynomial.

\begin{theorem} \label{invnonneg}
If $\ul{A}, \, \ol{A}$ are regular and $\ul{A}^{-1}, \, \ol{A}^{-1} \geq 0$ then $\imace{A}$ is regular and $$ \imace{A}^{-1} = [\ol{A}^{-1}, \ul{A}^{-1}] \geq 0.$$
%$\imace{A}$ is inverse nonnegative iff $\ul{A}^{-1}\geq0$ and $\ol{A}^{-1}\geq0$. Moreover, 
%$$ \imace{A}^{-1} = [\ol{A}^{-1}, \, \ul{A}^{-1} ].$$
\end{theorem}

\noindent \textbf{Summary.}
\begin{center}\renewcommand\arraystretch{1.3} 
\tabcolsep=1em
\begin{tabular}{@{}ll@{}}
\svhline
\emph{Problem} &  \emph{Complexity} \\
\hline
Computing the exact inverse of $\imace{A}$ &  NP-hard  \\
Is $\imace{A}$ inverse nonnegative? & strongly P \\
Computing the exact inverse of inverse nonnegative $\imace{A}$ & strongly P \\
\svhline
\end{tabular} 
\end{center}

\subsection{Solvability of a linear system}

Of course, before solving a linear system we might want to know, whether it is actually solvable. 
Considering solvability we should distinguish between two types of solvability. 
\begin{definition}
An interval linear system $\imace{A} x = \imace{b}$ is \emph{(weakly) solvable} if some system $Ax=b$, where $A \in \imace{A}, b \in \imace{b}$ is solvable.
\end{definition}
In another words, its solution set $\Sigma$ is not empty. Otherwise, we call the system \emph{insolvable}.

\begin{definition}
An interval linear system $\imace{A} x = \imace{b}$ is \emph{strongly solvable} if every system $Ax=b$, where $A \in \imace{A}, b \in \imace{b}$ is solvable.  
\end{definition}

The first definition is interesting for model checking. The second for system verification and automated proofs.

Checking whether an interval systems is solvable is an NP-hard problem \cite{kreinovich:complexity}. The sign coordinates of 
the orthant containing the solution can serve as a polynomial witness and existence of a solution can be verified by linear programming, hence this problem is NP-complete and checking unsolvability coNP-complete. 
The problem of deciding strong solvability is coNP-complete. It can be reformulated as checking insolvability of a certain linear system using the well known Farkas lemma, e.g., \cite{rohn:solvability}.

Sometimes, we look only for nonnegative solutions -- \emph{nonnegative solvability}. Checking whether an interval linear system has a nonnegative solution is weakly polynomial. We know the orthant in which the solution should lie. Therefore, we can get rid of the absolute values in Oettli-Prager theorem and apply linear programming. However, checking whether a system is nonnengative strongly solvable is still coNP-complete \cite{fiedler:linopt}. We summarize the results in the following table. 
\begin{theorem} Checking various types of solvability of $\imace{A} x= \imace{b}$ is of the following complexity.  
%\begin{table}[t]
\begin{center}\renewcommand\arraystretch{1.3} 
\tabcolsep=1em
\begin{tabular}{@{}lll@{}}
%&mid+rad error & absolute error & relative error & inverse relative error
\svhline 
 & weak & strong \\
\hline
solvability & NP-complete & coNP-complete \\ 
nonnegative solvability & P & coNP-complete \\
\svhline 
\end{tabular}
%\caption{}
%\end{table}
\end{center}
\end{theorem} \vspace{1em}

It is easy to see that an interval linear system $\imace{A} x = \imace{b}$ is insolvable if the matrix $[\imace{A} \ \imace{b}]$ has full column rank. That is why, we can use sufficient conditions for full column rank to check insolvability. 
Moreover, we can also use methods for computing enclosures. If we have some enclosure $\imace{x}$, then clearly a system $\imace{A} x = \imace{b}$ is unsolvable if $\imace{A} \imace{x} \cap \imace{b} = \emptyset$.
Many enclosure algorithms enable detection of insolvability. Gene\-rally speaking, 
they work in iterative stages and when we intersect enclosures of the solution set from the two subsequent stages and get an empty set, we know for sure that the system is insolvable.  
These methods are, for example, Gaussian elimination  \cite{hansen:ols}, Jacobi method \cite{moore:introduction}, Gauss-Seidel method \cite{moore:introduction}, subsquares method \cite{horacek:subsq}. \\

\noindent \textbf{Linear inequalities.}
Just for comparison, considering systems of interval linear inequalities, the problems of checking various types of solvability become much easier. The results are resumed in the following table \cite{fiedler:linopt}.
\begin{theorem}
Checking various types of solvability of $\imace{A} x \leq \imace{b}$ is of the following complexity. 
%\begin{table}[t]
\begin{center}\renewcommand\arraystretch{1.3} 
\tabcolsep=1em
\begin{tabular}{@{}lll@{}}
%&mid+rad error & absolute error & relative error & inverse relative error
\svhline
  & weak & strong \\
\hline
solvability &  NP-complete & P \\
nonnegative solvability & P & P \\
\svhline
\end{tabular}
%\caption{Complexity of checking various types of solvability for a system of interval linear inequalities}

\end{center}
%\end{table}
\end{theorem} \vspace{1em}
We also would like to mention an interesting nontrivial property of strong solvability of systems of interval linear inequalities. When a system $\imace{A} x \leq \imace{b}$ is strongly solvable (i.e., every $A x \leq b$ has a solution), then there exists a solution $x$ satisfying $A x \leq b$ for every $A\in\imace{A}$ and $b\in\imace{b}$ \cite{fiedler:linopt}. \\

\noindent \textbf{$\forall\exists$-solutions.}
Let us come back to interval linear systems. The traditional concept of a solution (Definition~\ref{dfSol}) employs existential quantifiers: $x$ is a solution if $\exists A\in\imace{A}$, $\exists b\in\imace{b}:Ax=b$. Nevertheless, in some applications, another quantification makes sense, too. In particular, $\forall\exists$ quantification was deeply studied \cite{Sha2002}. For illustration of complexity of such solution, we will focus on two concepts of solutions -- tolerance \cite{fiedler:linopt} and control solution \cite{fiedler:linopt,shary1992controlled}.

\begin{definition}\mbox{} \vspace{5pt}

\noindent A vector $x$ is a \emph{tolerance} solution of $\imace{A} x = \imace{b}$ if $\forall A\in\imace{A}$, $\exists b\in\imace{b}:Ax=b$. \vspace{3pt} \\ 
A vector $x$ is a \emph{control} solution of $\imace{A} x = \imace{b}$ if  $\forall b\in\imace{b}$, $\exists A\in\imace{A}:Ax=b$,
\end{definition}

Notice that a tolerance solution can equivalently be characterized as
$\{ A x \ | \ A \in \imace{A} \} \subseteq \imace{b}$
and  a control solution as
$\imace{b} \subseteq \{ A x \ | \ A \in \imace{A} \} .$

Both solutions can be described by a slight modification of Oettli-Prager theorem (one sign change in Oettli-Prager formula) \cite{fiedler:linopt}. 
\begin{theorem}
Let us have a system $\imace{A} x = \imace{b}$, then $x$ is
\begin{itemize}
\item a tolerance solution if it satisfies $ | A_c x - b_c | \leq -\Delta |x| + \delta. $ \vspace{3pt}
\item a control solution if it satisfies $ | A_c x - b_c | \leq \Delta |x| - \delta. $
\end{itemize}
\end{theorem}
In case of tolerance solution this change makes checking whether a systems has this kind of solution decidable in weakly polynomial time. In the case of control solution the decision problem stays NP-complete. The same complexity holds for a problem of deciding whether an interval linear systems has a tolerance (polynomial) or control solution (NP-complete) \cite{kreinovich:complexity}. \\
\newpage
\noindent \textbf{Summary.}
\begin{center}\renewcommand\arraystretch{1.3} 
\tabcolsep=1em
\begin{tabular}{@{}ll@{}}
\svhline
 \emph{Problem} & \emph{Complexity} \\
\hline

Is $\imace{A}x=\imace{b}$ solvable? &  NP-complete  \\
Is $\imace{A}x=\imace{b}$ strongly solvable? & coNP-complete \\
Is $\imace{A}x=\imace{b}$ nonnegative solvable? &  P  \\
Is $\imace{A}x=\imace{b}$ nonnegative strongly solvable? & coNP-complete \\
Is $\imace{A}x\leq\imace{b}$ solvable? &  NP-complete  \\
Is $\imace{A}x\leq\imace{b}$ strongly solvable? & P \\
Is $\imace{A}x\leq\imace{b}$ nonnegative solvable? &  P  \\
Is $\imace{A}x\leq\imace{b}$ nonnegative strongly solvable? & P \\
Is $x$ a tolerance solution of $\imace{A}x=\imace{b}$? & P \\
Is $x$ a control solution of $\imace{A}x=\imace{b}$? & NP-complete \\
Does $\imace{A}x=\imace{b}$ have a tolerance solution? & P \\
Does $\imace{A}x=\imace{b}$ have a control solution? & NP-complete \\
\svhline

\end{tabular} 
\end{center}

%================================================================================================
\subsection{Determinant}
Determinants of interval matrices are not often studied. However, we included this section for completeness. 
\begin{definition}
A determinant of $\imace{A}$ is defined as $\det(\imace{A}) = [\ul{d}, \ol{d}]$, where $$\ul{d} = \min \{ \det(A) \ | \ A \in \imace{A}\},$$ 
$$\ol{d} = \max \{ \det(A) \ | \ A \in \imace{A}\}.$$ 
\end{definition}
Its bounds can be computed from $2^{n^2}$ boundary matrices $A_{ij} \in \{ \ul{A}_{ij}, \ol{A}_{ij}\}$ for $i,j = 1, \ldots, n$. 
We have the following theoretical result \cite{rohn:checking}.
\begin{theorem}
Computing interval determinant of $\imace{A} = [A - E, A + E]$, where $A$ is rational nonnegative is NP-hard. 
\end{theorem}
It is intractable even in this simplified case.
For interesting relations to eigenvalues and singularity see \cite{rohn:checking}. \\

\noindent \textbf{Summary.}
\begin{center}\renewcommand\arraystretch{1.3} 
\tabcolsep=1em
\begin{tabular}{@{}ll@{}}
\svhline
 \emph{Problem} & \emph{Complexity} \\
\hline

Computing $\ul{\det} (\imace{A})$  &   NP-hard  \\
Computing $\ol{\det} (\imace{A})$  &   NP-hard  \\
\svhline 
\end{tabular}
\end{center}

%================================================================================================
\subsection{Eigenvalues}
First, we briefly start with general matrices, then we continue with the symmetric case. 
Checking singularity of $\imace{A}$ can be polynomially reduced to checking whether $0$ is an eigenvalue of some matrix $A \in \imace{A}$.
As we saw in section \ref{eigen} checking whether $\lambda$ is an eigenvalue of some matrix $A \in \imace{A}$ is NP-complete problem.
Surprisingly, checking for eigenvectors can be done efficiently  \cite{rohn:sineig}. It is strongly polynomial.

How it is with Perron theory?
An interval matrix $\imace{A}\in\IR^{n\times n}$ is \emph{nonnegative irreducible} if every $A\in \imace{A}$ is nonnegative irreducible. For Perron vectors (positive vectors corresponding to the dominant eigenvalues), we have
the following result \cite{Roh2005c}.
\begin{theorem}
Let $\imace{A}$ be nonnegative irreducible. Then the problem of deciding whether $x$ is a Perron eigenvector of some matrix in $\imace{A}$ is strongly polynomial.
\end{theorem}
For the sake of simplicity we mentioned only some results considering eigenvalues of a general matrix $\imace{A}$. We will go into more detail with symmetric matrices, where their eigenvalues are real. 

\begin{definition}
Let $\imace{A}\in\IR^{n\times n}$ with $\Delta, \,A_c$ symmetric. Then the corresponding \emph{symmetric interval matrix} is defined as a subset of symmetric matrices in $\imace{A}$, that is, 
$$
\smace{A}:= \{A \in \imace{A}:A=A^T\}.
$$
\end{definition}

For a symmetric $A\in\R^{n\times n}$, we use $\lmin(A)$ and $\lmax(A)$ for its smallest and largest eigenvalue, respectively. For a symmetric interval matrix, we define the smallest and largest eigenvalues respectively as
\begin{eqnarray}
\lmin(\smace{A})&:=\min\{\lmin(A): A\in\smace{A}\}, \nonumber \\
\lmax(\smace{A})&:=\max\{\lmax(A): A\in\smace{A}\}.\nonumber
\end{eqnarray}

Even if we consider the symmetric case some problems remain intractable \cite{kreinovich:complexity,rohn:checking}. We are yet able to prove the 
hardness results, since it is difficult to find a proper polynomial witness.

\begin{theorem}
On a class of problems with $A_c\in\Q^{n\times n}$ symmetric positive definite and entrywise nonnegative, and $\Delta=E$, the following problems are intractable 
\begin{itemize}
\item
checking whether $0$ is an eigenvalue of some matrix $A\in\smace{A}$ is NP-hard, \vspace{3pt}
\item
checking $\lmax(\smace{A})\in(\ul{a},\ol{a})$ for a given open interval $(\ul{a},\ol{a})$ is coNP-hard.
\end{itemize}
\end{theorem}

However, there are some known subclasses for which the eigenvalue range or at least one of the extremal eigenvalues can be determined efficiently \cite{Hla2015b}:
\begin{itemize}
\item
If $A_c$ is \emph{essentially non-negative}, i.e., $(A_c)_{ij}\geq 0\ \forall i \neq j$, then $\lmax(\smace{A})=\lmax(\ol{A})$.\vspace{3pt}
\item
If $\Delta$ is \emph{diagonal}, then $\lmin(\smace{A})=\lmin(\ul{A})$ and $\lmax(\smace{A})=\lmax(\ol{A})$.
\end{itemize}

In contrast to the extremal eigenvalues $\lmin(\smace{A})$ and $\lmax(\smace{A})$, the largest of the minimal eigenvalues and the smallest of the largest eigenvalues, \begin{eqnarray}
&\max\{\lmin(A): A \in \smace{A}\}, \nonumber\\
&\min\{\lmax(A): A \in \smace{A}\}, \nonumber
\end{eqnarray}
can be computed with an arbitrary precision in polynomial time by using semidefinite programming \cite{JauHen2005}.
As in the general case, checking whether a given vector $0\not=x\in\R^n$ is an eigenvector of some matrix in $\smace{A}$ is a polynomial time problem. Nevertheless, strong polynomiality has not been proved yet.

We already know that computing exact bounds on many problems with interval data is intractable. Since we can do no better, we can inspect the hardness of various approximations of their solutions. While doing this we use the following assumption: \emph{
Throughout this section, we consider a computational model, in which the exact eigenvalues of rational symmetric matrices are polynomially computable}.

The table below from \cite{Hla2015b} summarizes the main results.
We use the symbol $\infty$ in case there is no finite approximation factor with polynomial complexity.

\begin{theorem}
Approximating the extremal eigenvalues of $\smace{A}$ is of the following complexity.
%\begin{table}[t]
\begin{center}\renewcommand\arraystretch{1.3} 
\tabcolsep=1em
{\begin{tabular}{@{}llll@{}}
%&mid+rad error & absolute error & relative error & inverse relative error
\svhline
& abs.\ error & rel.\ error & inverse rel.\ error \\
\hline
NP-hard with error &  any & $<1$ & $1$ \\
polynomial with error  & $\infty$ & $1$ & $2$ \\
\svhline
\end{tabular}}
%\caption{}
%\end{table}
\end{center}
\end{theorem} \vspace{1em}

\noindent The table below gives analogous results for the specific case of approximating 
$\lmax(\smace{A})$ when $A_c$ is positive semi-definite.

\begin{theorem}
Approximating the extremal eigenvalues of $\smace{A}$ with $A_c$ rational
positive semi-definite is of the following complexity.
%\begin{table}[t]
\begin{center}\renewcommand\arraystretch{1.3} 
\tabcolsep=1em
{\begin{tabular}{@{}llll@{}}
%&mid+rad error & absolute error & relative error & inverse relative error
\svhline
& abs.\ error & rel.\ error & inverse rel.\ error\\
\hline
NP-hard with error  & any & $1/(32n^4)$ &  $1/(32n^4)$ \\
polynomial with error  & $\infty$ & $1/3$ & $1/3$ \\
\svhline
\end{tabular}}
%\caption{}
%\end{table}
\end{center}
\end{theorem} \vspace{1em}
The tables sums up the generalized idea behind several theorems on computing extremal eigenvalues. For more information and formal details see
\cite{Hla2015b}.

At the end of this subsection we mention spectral radius.

\begin{definition}
Let $\imace{A}\in\IR^{n\times n}$, we define the range of \emph{spectral radius} naturally as
$$ \rho(\imace{A})=\{\rho(A): A\in\imace{A}\}.$$
\end{definition}

Notice that $\rho(\imace{A})$ is a compact real interval due to continuity of eigenvalues. Similarly we define spectral radius for $\smace{A}$.

Complexity of computing $\ol{\rho(\imace{A})}$ is an open problem (as Schur stability is; see Section~\ref{ssStab}), and, to the best of our knowledge, complexity of computing $\ul{\rho(\imace{A})}$ has not been investigated yet.

Anyway, the following gives polynomially solvable subclasses:
\begin{itemize} 
\item
If $\ul{A}\geq0$, then $\rho(\imace{A})=[\rho(\ul{A}),\rho(\ol{A})]$. \vspace{3pt}
\item
If $\imace{A}$ is diagonal, then 
$\rho(\imace{A})=[\max_i\min_{a\in\imace{a}_{ii}}|a|,\,\max_i\{|\ul{a}_{ii}|,|\ol{a}_{ii}|\}]$.
\end{itemize}

\noindent \textbf{Summary.}
\begin{center}\renewcommand\arraystretch{1.3} 
\tabcolsep=1em
\begin{tabular}{@{}ll@{}}
\svhline
 \emph{Problem} & \emph{Complexity} \\
\hline
Is $\lambda$ eigenvalue of some $A \in \imace{A}?$ & NP-complete \\
Is $x$ eigenvector of some $$ $A \in \imace{A}$? &  strongly P  \\ 
Is $x$ Perron vector of nonnegative irreducible $\imace{A}$? & strongly P \\
Is 0 eigenvalue of some $A \in \smace{A}?$ & NP-hard \\
Is $x$ eigenvector of some $$ $A \in \smace{A}$? &  P  \\ 
Does $\lmax(\smace{A})$ belong to a given open interval? & coNP-hard \\
Computing $\ol{\rho(\imace{A})}$ & ? \\
Computing $\ul{\rho(\imace{A})}$ & ? \\
Computing exact bounds on $\rho(\imace{A})$ with $\imace{A}$ nonnegative &  strongly P \\
Computing exact bounds on $\rho(\imace{A})$ with $\imace{A}$ diagonal & strongly P \\
\svhline
\end{tabular}
\end{center}

%==============================================================================
\subsection{Positive definitness and semidefiniteness}

We should not leave out mentioning the positive definiteness and semidefiniteness. Here without the loss of the generality symmetric matrices are of the only interest. We distinguish between weak and strong definiteness.

\begin{definition}
A symmetric interval matrix $\smace{A}$ is weakly positive (semi)definite if some $A \in \smace{A}$ is positive (semi)definite.  
\end{definition}

\begin{definition}
A symmetric interval matrix $\smace{A}$ is strongly positive (semi)definite if every $A \in \smace{A}$ is positive (semi)definite.  
\end{definition}

%One of the necessary and sufficient conditions for strong positive definiteness is formulated in the following interesting theorem \cite{rohn1994positive}.
%
%\begin{theorem}
%A symmetric interval matrix $\smace{A}$ is strong positive definite if and only if $\imace{A}$ is regular and contains at least one positive definite matrix.
%\end{theorem}

Checking positive definiteness \cite{Roh1994} and semidefiniteness \cite{Nem1993} are both coNP-hard according to the two following theorems.

\begin{theorem}
Checking strong positive semidefiniteness of $\smace{A}$ is co-NP-hard on a class of problems with $A_c\in\Q^{n\times n}$ symmetric positive definite and entrywise nonnegative, and $\Delta=E$.
\end{theorem}

\begin{theorem}
Checking strong positive definiteness of $\smace{A}$ is co-NP-hard on a class of problems with $A_c\in\Q^{n\times n}$ symmetric positive definite and entrywise nonnegative, and $\Delta=E$.
\end{theorem}

Considering positive definiteness, we have some sufficient conditions that can be checked polynomially \cite{rohn1994positive}. 

\begin{theorem}
An interval matrix $\smace{A}$ is strongly positive definite if at least one of the following condition holds
\begin{itemize}
\item $\lambda_n(A_c)>\rho(\Delta)$, \vspace{3pt}
\item $A_c$ is positive definite and $\rho(|(A_c)^{-1}|\Delta)<1$.
\end{itemize} 
\end{theorem}
The second condition can be reformulated as $\smace{A}$ being regular and $A_c$ positive definite.
If the first condition holds with $\geq$ then $\smace{A}$ is strongly positive semidefinite. 

In contrast to checking strong positive definiteness, weak positive definiteness can be checked in polynomial time by using semidefinite programming \cite{JauHen2005}; this polynomial result holds also for a more general class of symmetric interval matrices with linear dependencies \cite{Hla2016pa}.
For positive semidefiniteness it needn't be the case since semidefinite programming methods work only with some given accuracy.\\

\noindent \textbf{Summary.}
\begin{center}\renewcommand\arraystretch{1.3} 
\tabcolsep=1em
\begin{tabular}{@{}ll@{}}
\svhline
 \emph{Problem} & \emph{Complexity} \\
\hline
Is $\smace{A}$ strongly positive definite? & coNP-hard \\
Is $\smace{A}$ strongly positive semidefinite? & coNP-hard \\
Is $\smace{A}$ weakly positive definite? & P \\
Is $\smace{A}$ weakly positive semidefinite? & ? \\
\svhline
\end{tabular}
\end{center}

%====================================================================================
\subsection{Stability}\label{ssStab}
The last section is dedicated to an important and more practical problem -- deciding a stability of a matrix. There are many types of stabilities. For illustration, we chose two of them -- Hurwitz and Schur.

\begin{definition}
An interval matrix $\imace{A}$ is \emph{Hurwitz stable} if every $A\in\imace{A}$ is Hurwitz stable (i.e., all eigenvalues have negative real parts).  
\end{definition}

Similarly, we define Hurwitz stability for symmetric interval matrices. Due to their relation to positive definiteness ($\smace{A}$ is Hurwitz stable if $-\smace{A}$ is positive definite) we could presume that the problem is coNP-hard. It is so, even if we limit ourselves to a special case \cite{Roh1994}.

\begin{theorem}
Checking Hurwitz stability of a symmetric interval matrix $\smace{A}$ is coNP-hard on a class of problems with $A_c\in\Q^{n\times n}$ symmetric Hurwitz stable and entrywise nonpositive, and $\Delta=E$.
\end{theorem}

For general matrices, coNP-hardness holds as well. The problem is still coNP-hard even if we limit the number of interval coefficients in our matrix \cite{Nem1993}.

\begin{theorem}
Checking Hurwitz stability of $\imace{A}$ is co-NP-hard on a class of interval matrices with intervals in the last row and column only.
\end{theorem}

Likewise, as for checking regularity, also checking Hurwitz stability of $\imace{A}$ can not be done by checking stability of matrices of type $A_{yz}$ (for reductions of other properties see \cite{GarAdm2016a}).
On the other hand, it can be checked in this way for $\smace{A}$. For more discussion and historical context see \cite{kreinovich:complexity}  or \cite{rohn:handbook}. As sufficient conditions we can use conditions for positive definiteness applied to $-\imace{A}$. For more sufficient conditions see e.g., \cite{mansour1989robust}.

\begin{definition}
An interval matrix $\imace{A}$ is \emph{Schur stable} if every $A\in\imace{A}$ is Schur stable (i.e., $\rho(A)<1$).  
\end{definition}

In a similar way, we define Schur stability for symmetric interval matrices. For general interval matrices, complexity of checking Schur stability is an open problem, however, for the symmetric case the problem is intractable \cite{Roh1994}. 

\begin{theorem}\label{thmSchurSymStab}
Checking Schur stability of $\smace{A}$ is coNP-hard on a class of problems with $A_c\in\Q^{n\times n}$ symmetric Schur stable and offdiagonal entries nonpositive, and $\Delta=E$.
\end{theorem} 

\noindent \textbf{Summary.}
\begin{center}\renewcommand\arraystretch{1.3} 
\tabcolsep=1em
\begin{tabular}{@{}ll@{}}
\svhline
 \emph{Problem} & \emph{Complexity} \\
\hline
Is $\imace{A}$ Hurwitz stable? & coNP-hard \\
Is $\smace{A}$ Hurwitz stable? & coNP-hard \\
Is $\imace{A}$ Schur stable? & ? \\
Is $\smace{A}$ Schur stable? & coNP-hard \\
\svhline
\end{tabular}
\end{center}

%================================================================================
\subsection{Further topics}
\label{further}
Due to the limited space, we had to omit many interesting topics. 
We touched only briefly the complexity issues of interval linear inequalities, but there are more results; see, e.g., \cite{fiedler:linopt,Hla2015a}.
We did not discussed complexity of computing the range of polynomials over intervals \cite{kreinovich:complexity}, too.
In short, we mention two particular problems:
\begin{itemize}
\item
\emph{Matrix power.} Computing the exact bounds on second power of the matrix $\imace{A}^2$ is strongly polynomial (just by evaluating by interval arithmetic), but computing the cube $\imace{A}^3$ turns out to be NP-hard \cite{KoshKre2005}.  
\item
\emph{Matrix norm.} Computing the range of $\|A\|$ when $A\in\imace{A}$ is a trivial task for vector $\ell_p$-norms applied on matrices (including Frobenius norm or maximum norm) or for induced $1$- and $\infty$-norms. On the other hand, determining the largest value of the spectral norm $\|A\|_2$ (the largest singular value) subject to $A\in\imace{A}$ is NP-hard \cite{Nem1993}.
\end{itemize}

\section{Summary}
In this work we explored the fundamental problems of interval linear algebra. Our goal was to:
\begin{itemize}
\item provide a basic introduction to interval linear algebra
\item answer elementary computational complexity questions linked with interval linear algebra 
\item discuss the computational complexity of the basic problems
\item explain the relations between these problems  
\item mention relaxations or special classes of these problems that are easily decidable or there exist polynomial algorithms solving them
\item provide a basis for further reading and research 
\end{itemize}

At this place we also would like to apologize to those whose results are not mentioned in this work. There are many great achievements, however this work can unfortunately consume only limited amount of space. We provide links to the literature, where you can find much more of them.

\subsection*{Acknowledgement.}
J.~Hor\'a\v cek and M.\ Hlad\'{\i}k were supported by GA\v{C}R grant P402/13-10660S. 
M.~\v{C}ern\'y was supported by the GA\v{C}R grant 16-00408S.

\end{document}